\documentclass[a4j]{article}

\usepackage{jsaisig}

\usepackage[dvipdfm]{graphicx}

\usepackage{amsmath}
\usepackage{amsfonts}
\usepackage{amssymb}
\usepackage{amsthm}

\newtheorem{thm}{Theorem}

\newtheorem{lem}{Lemma}
\newtheorem{cor}{Corollary}
\newtheorem{de}{Definition}

\begin{document}

% 和文タイトル
\title{For the 9th conference of Japanese Society for Artificial Intelligence,
 Special Interest Group on Data Mining and Statistical Mathematics
     \\     (JSAI SIG-DMSM) in Kyoto 2009, March 3,4.}

% 英文タイトル
\etitle{Graph polynomials and approximation of partition functions with
Loopy Belief Propagation} 

% 著者名: 

\author{Yusuke Watanabe%
	\thanks{Address: The Institute of Statistical Mathematics,
4-6-7, Minami-Azabu, Minato-Ku, Tokyo, 106-8569
E-mail:watay@ism.ac.jp } , \quad
	Kenji Fukumizu${}^*$
	}

% 所属
\affiliation{%
%	\afil{1} The Institute of Statistical Mathematics \\ 
}

\abstract{
The Bethe approximation, or loopy belief propagation
algorithm, is a successful method for approximating
partition functions of probabilistic models associated with graphs.

Chertkov and Chernyak derived an interesting formula
called ``Loop Series Expansion'',
which is an expansion of the partition function.
The main term of the series is the Bethe approximation while the other terms are 
labeled by subgraphs called generalized loops.

In our recent paper, we derive the loop series expansion in 
form of a polynomial with coefficients positive integers, and extend the
result to the expansion of marginals.
In this paper, we give more clear derivation of the results and discuss 
the properties of newly introduced polynomials.
}
\maketitle
\thispagestyle{empty}

%%%%%%%%%%%%%%%%%%%%%%%%%%%%%%%%%%%%%%

\section{Introduction}
A Markov random field (MRF) associated with a graph is given by a joint
probability distribution over a set of variables.
In the associated graph,
the nodes represent variables and the edges represent probabilistic
dependence between variables. 
The joint distribution is often given in an unnormalized form,
and the normalization factor of a MRF is called a partition function.

Computation of the partition function and the mar-ginal distributions 
of a MRF with discrete variables is in general computationally
intractable for a large number of variables, and some approximation method is required.
Among many approximation methods, the Bethe approximation %\cite{Bethe}
has attracted renewed interest of computer scientists;
it is equivalent to Loopy Belief Propagation
(LBP) algorithm \cite{empiricalstudy,GBP}, which has been successfully used for 
many applications such as error correcting codes, 
inference on graphs, image processing, and so on
\cite{LDPC,Turbo,LowLevelVision}.
If the associated graph is a tree,
the algorithm computes the exact value, not an approximation \cite{Pearl}.

%
%近似の性能について理論的に解明する必要があるということの説明
The performance of this approximation is surprisingly well for many
applications even if the graph has many cycles. 
If the graph has one cycle,
the behavior of the algorithm is well understood,
and maximum marginal assignment of the approximation is known to be
exact \cite{1loop}.
On the other hand, if the graph has many cycles, little analysis have
done in the connection with the topological structure of the underlying graph.
Theoretical analysis of the approximation is important both for improving the
algorithm and extending it to wide range of applications.

%C.C.の業績
Chertkov and Chernyak \cite{LoopPRE,Loop} give a
formula called loop series expansion, which expresses the partition
function in terms of a finite series.
The first term is the Bethe approximation, and the others are labeled by
so-called generalized loops.
The Bethe approximation can be corrected with this formula.
This expansion highlights the connection between the accuracy of the
approximation and the topology of the graph.

In our recent paper \cite{watanabe}, we derive the formula in terms
of message passing scheme with diagrams based on the method called propagation
diagrams.
We also showed that the true marginals can be also expanded around
the approximated marginals. 

%この論文の内容
In this paper we give a simple and easy derivation of the expansion of
the partition function and the marginal distributions. 
Similar but a different approach is found in \cite{Sudderth}.
We also discuss the properties of the bivariate polynomial which is
introduced in \cite{watanabe}.
Since this polynomial represents the ratio of the true partition function
and the Bethe approximated partition function, investigation of it is
important for understanding the relation between the graph topology and
the approximation performance.

%%%%%%%%%%%%%%%%%%%%%%%%%%%%%%%%%%%%%%
\section{Bethe approximation and LBP algorithm}
In this section, we review the definitions and notations of Markov Random
Field (MRF) and loopy belief propagation algorithm.
\subsection{Pairwise Markov random field}
%考えるモデルの定義
We introduce a probabilistic model considered in this paper,
MRF of binary states with pairwise interactions.
Let $G:=(V,E)$ be a connected undirected graph, where
$V=\{1,\ldots,N \}$ is a set of nodes and 
$E\subset \{(i,j); 1\leq i < j \leq N \}$
is a set of undirected edges.
We abbreviate undirected edges $(i,j)$ to $ij=ji$.
Each node $i \in V$ is associated with a binary space $\chi_i=\{\pm 1\}$.
We make a set of directed edges from $E$: $\vec{E}=\{(i,j),(j,i);(i,j)\in E
\}$.
The neighbors of $i$ is denoted by
$N(i) \subset V$, and $d_i= |N(i)|$ is called the degree of $i$.
A joint probability distribution on the graph $G$ is given by the form:
\begin{eqnarray}
p(x)=\frac{1}{Z} \prod_{ij\in E} \psi_{ij}(x_i,x_j) 
\prod_{i \in V} \phi_{i}(x_i), \label{MRF}
\end{eqnarray}
where $\psi_{ij}(x_i,x_j):\chi_i \times \chi_j \rightarrow
\mathbb{R}_{> 0}$ and  
$\phi_i : \chi_i \rightarrow \mathbb{R}_{> 0}$ are positive
functions called compatibility functions.
The normalization factor $Z$ is called the partition function.
A set of random variables which has a probability distribution in the
form of (\ref{MRF}) is called a Markov random
field (MRF) or an undirected graphical model on the graph $G$.

Without loss of generality, univariate compatibility functions $ \phi_{i}$ can 
be neglected because they can be included in bivariate compatibility functions $\psi_{ij}$.
This operation does not affect the Bethe approximation and the LBP algorithm
given below; we assume it as per the following.

\subsection{Loopy belief propagation algorithm}
The LBP algorithm computes the Bethe approximation of
the  partition function and the marginal distribution of each node
with the message passing method \cite{Pearl,empiricalstudy,GBP}.
This algorithm is summarized as follows.
%BPのアルゴリズム
\begin{enumerate}
\item Initialization: \\
For all $(j,i) \in \vec{E}$, the message from $i$ to $j$ is a vector
 $m_{(j,i)}^{0} \in \large{\mathbb{R}}^{2} $. Initialize as
\begin{eqnarray}
m_{(j,i)}^{0}(x_j)=1  \quad \forall x_j \in \chi_j .
\end{eqnarray}

\item Message Passing: \label{MessagePassing}\\
For each $t=0,1,\ldots$, update the messages by
\begin{equation}
m^{t+1}_{(j,i)}(x_j)=\omega \sum_{x_i \in \chi_i} 
\hspace{-1mm}\psi_{ji}(x_j,x_i)
\hspace{-2mm}
\prod_{k \in N(i) \backslash  \{j\} } 
\hspace{-4mm}
m^t_{(i,k)}(x_i), \label{BPupdate}
\end{equation}
until it converges. Finally we obtain
$\{m^{*}_{(j,i)}\}$.

\item 
Approximated marginals and the partition function are computed
by the following formulas:
\begin{equation}
\hspace{-23mm}
b_i(x_i):= \omega \prod_{j \in N(i)}m^{*}_{(i,j)}(x_i), \label{1marginal}
\end{equation}
\vspace{-6mm}
\begin{multline}
b_{ji}(x_j,x_i):= 
\omega  \psi_{ji}(x_j,x_i)  
\hspace{-2mm}
\prod_{k \in
 N(j)\backslash \{i\}} 
\hspace{-2mm} m^*_{(j,k)}(x_j) \\
\prod_{k' \in N(i) \backslash \{j\}} 
\hspace{-2mm} m^*_{(i,k')}(x_i) ,\label{2marginal}
\end{multline}
\vspace{-4mm}
%Bethe近似分配関数
\begin{multline}
\log Z_B := \sum_{ji \in {E}}\sum_{x_j, x_i}
 b_{ji}(x_j,x_i)\log\psi_{ji}(x_j,x_i)   \\
 - \sum_{ji \in {E}}\sum_{x_j x_i} b_{ji}(x_j,x_i)\log
 b_{ji}(x_j,x_i) \nonumber  \\
+ \sum_{i \in V}(d_i-1)\sum_{x_i}b_i(x_i)\log b_i(x_i), \label{Bethe}
\end{multline}
where $\omega$ are appropriate normalization constants,
$b_i$ are called beliefs, and $Z_B$ is called the Bethe approximation of the partition function.
\end{enumerate}
In step 2, there is ambiguity as to the order of
updating the messages. 
We do not specify the order, because the fixed points of LBP algorithm
do not depend on its choice.

We normalize as
$\sum_{x_i,x_j}b_{ji}(x_j,x_i)=\sum_{x_i}b_{i}(x_i)=1$,
so that 
the relation
$\sum_{x_i}b_{ji}(x_j,x_i)=b_j(x_j)$ is always satisfied for all $ij \in
E$.

%Note that this LBP algorithm does not necessarily converge, and there
%may be more than one fixed points unless the interactions are sufficiently
%weak \cite{Gibbsmeasure}. 

%%%%%%%%%%%%%%%%%%%%%%%%%%%%%%%%%%%%%%%%%%%%%%%%%%%%%%%%
\section{Derivation of the Loop Series Expansion}
%section3
%このセクションでやること
In this section, we prove the loop series expansion formula of the partition
function and marginals.

\subsection{Expansion of partition functions}
First, we prove the following identity which plays a key role in our
derivation.
We define a set of polynomials $\{f_n(x)\}_{n=0}^{\infty}$ 
inductively by the relations $f_0(x)=1,f_1(x)=0$ and $f_{n+1}(x)=x
f_n(x) + f_{n-1}(x)$.
This polynomials are transformations of the Chebyshev polynomials of the second
kind.
%定理１
\begin{thm}
Let $\{\xi_{i} \}_{i \in V}$ and $\{\beta_{ij} \}_{ij \in E}$ be sets of
 free variables associated to nodes and edges.  
Then,
\begin{equation}
\begin{split}
\sum_{x_1,\ldots,x_N=\pm 1}
\prod_{ij \in E}
(1+x_{i}x_{j}\beta_{ij}\xi_{i}^{-x_i}\xi_{j}^{-x_j})
\prod_{i \in V}
\frac{\xi_{i}^{x_i}}{\xi_{i}+\xi_{i}^{-1}} \\ 
=
\sum_{s \subset E}
\prod_{ij \in s} \beta_{ij}
\prod_{i \in V} f_{d_i(s)}(\xi_{i}-\xi_{i}^{-1}), 
\end{split}\label{thm1}
\end{equation}
where $d_{i}(s)$ is the degree of node $i$ in the subgraph induced by 
an edge set $s$.
\end{thm}
%証明
\begin{proof}
\begin{equation*}
\begin{split}
(L.H.S)&=\sum_{\{x_i\}}\sum_{s \subset E}
\prod_{ij \in s} x_i x_j \beta_{ij} \xi_{i}^{-x_i}\xi_{j}^{-x_j} 
\prod_{i \in V} \frac{\xi_{i}^{x_i}}{\xi_{i}+\xi_{i}^{-1}} \\
&=
\sum_{s \subset E}
\prod_{ij \in s} \beta_{ij}
\prod_{i \in V} \sum_{x_i=\pm 1}(-x_i  \xi_{i}^{-x_i})^{d_{i}(s)}
\frac{\xi_{i}^{x_i}}{\xi_{i}+\xi_{i}^{-1}} \\ 
&=
\sum_{s \subset E}
\prod_{ij \in s} \beta_{ij}
\prod_{i \in V} 
\frac{-(-\xi_{i})^{-d_i(s)+1}+\xi_{i}^{d_i(s)-1}}{\xi_{i}+\xi_{i}^{-1}}.
\end{split}
\end{equation*}
On the other hand, by the definition of $f_n$
\begin{equation}
f_n(\xi-\xi^{-1})=
\frac{\xi^{n-1}-(-\xi)^{-n+1}}{\xi_{}+\xi_{}^{-1}}. \label{fn}
\end{equation}
\end{proof}

%恒等式の証明終わり
% 

%Lemma 1
Secondly, we give a relation between the true partition function and the
Bethe approximation.
\begin{lem}
\begin{equation}
\frac{Z}{Z_B} =
\sum_{\{x_i\}}
\prod_{ij \in E}
\frac{b_{ij}(x_i,x_j)}{b_i(x_i)b_j(x_j)}
\prod_{i \in V}
b_i(x_i) \label{lem1}
\end{equation}
\end{lem}
%証明
\begin{proof}
If we write the normalization terms explicitly, (\ref{1marginal}) and
 (\ref{2marginal}) are rewritten as
\begin{equation}
\begin{split} \label{lem1p1}
&b_i(x_i)= c_i^{-1} \prod_{j \in N(i)}m^{*}_{(i,j)}(x_i), 
\hspace{10mm} {} \\
\vspace{-4mm}
&b_{ji}(x_j,x_i)= 
c_{ji}^{-1} \psi_{ji}(x_j,x_i)  
\hspace{-5mm}
\prod_{k \in
 N(j)\backslash \{i\}} 
\hspace{-5mm} m^*_{(j,k)}(x_j) 
\hspace{-5mm}
\prod_{k' \in N(i) \backslash \{j\}} 
\hspace{-6mm} m^*_{(i,k')}(x_i) .
\end{split}
\end{equation}
By the definition of the Bethe approximation of partition function, it is easy to see that
\begin{equation}
Z_B=\prod_{ij \in E}
\frac{c_{ij}}{c_i c_j}
\prod_{i \in V}c_i.
\end{equation}
Then, we use (\ref{lem1p1}) again, the right hand side of (\ref{lem1}) is equal to 
$\frac{Z}{Z_B}$.
\end{proof}

%%Lemma1終了。
Finally, we give the loop series expansion formula of the partition
function \cite{watanabe,LoopPRE,Loop}.
%定理2の証明
\begin{thm} \label{thmloop}
Let 
\begin{equation}
\gamma_{i}:=\frac{b_i(1)-b_i(-1)}{\sqrt{b_i(1)b_i(-1)}}, \label{defgamma}
\end{equation}
and
\begin{equation}
\beta_{ij}:=\frac{b_{ij}(1,1)b_{ij}(-1,-1)-b_{ij}(1,-1)b_{ij}(-1,1)}
{\sqrt{b_i(1)b_i(-1)}\sqrt{b_j(1)b_j(-1)}}.  \label{defbeta}
\end{equation}
Then, the following formula holds.
\begin{gather}
Z=Z_B
\sum_{s \subset E } r(s) ,\label{LSE} \\
r(s):=\prod_{ij \in s}\beta_{ij} 
\prod_{i \in V} f_{d_{i}(s)}(\gamma_{i}).
\end{gather}
\end{thm}
Before accomplishing the proof, let us consider the meaning of the
theorem. Equation (\ref{defgamma}) states that $\gamma_{i}$ is related to
the bias of the approximated marginal $b_i(x_i)$, and $\gamma_{i}=0$ if and only if $b_{i}(1)=b_{i}(-1)$.
Equation (\ref{defbeta}) states that $\beta_{ij}$ is related to the 
correlation by $b_{ij}$.
It is easy to see that $|\beta_{ij}| \leq 1$ and $\beta_{ij}=0$ if and
only if $b_{ij}(x_i,x_j)=b_{i}(x_i)b_{j}(x_j)$.

The definitions and properties of these quantities $\beta_{ij}$ and
$\gamma_i$, in the context of message passing procedures, are found in \cite{watanabe}.

In (\ref{LSE}), the summation runs over all subsets of $E$ including
$s=\phi$. As $f_1(x)=0$, $s$ makes a contribution to the sum only if $s$ does not have a node
of degree one in $s$. Such a subgraph $s$ is called a generalized loop
\cite{LoopPRE,Loop} or a closed graph \cite{NaglePR,NagleJC}.
In the case that $G$ is a tree,
there is no generalized loops, therefore $Z=Z_B$. 
This is an alternative proof of the well-known fact
\cite{Pearl}.
 
If $\beta_{ij}$ an $\gamma_i$ are sufficiently small, the first term
$r(\phi)=1$ is mainly contribute to the sum, and $Z_B$ is close to $Z$.

%定理2の証明
\begin{proof}[Proof of theorem \ref{thmloop}]
For a given $b_{ij}(x_i,x_j)$ which satisfies the 
normalization condition
 $\sum_{x_i,x_j}b_{ij}(x_i,x_j)=1$,
we can always choose $\xi_i,\xi_j,\beta_{ij}'$ to satisfy 
\begin{equation}
b_{ij}(x_i,x_j)=\frac{1}{(\xi_i+\xi_i^{-1})(\xi_j+\xi_j^{-1})}
(\xi_i^{x_i}\xi_j^{x_j}+\beta_{ij}' x_i x_j).
\end{equation}
From (\ref{defbeta}), we see that $\beta_{ij}'=\beta_{ij}$.
Using
\begin{equation}
b_i(x_i)=\sum_{x_j}b_{ij}(x_i,x_j)=\frac{\xi_i^{x_i}}{\xi_i+\xi_i^{-1}}
\end{equation}
and (\ref{defgamma}), we have $\gamma_{i}=\xi_i-\xi_i^{-1}$.
Notice that 
\begin{equation}
\frac{b_{ij}(x_i,x_j)}{b_i(x_i)b_j(x_j)}=1+x_i x_j \beta_{ij}\xi_i^{-x_i}\xi_j^{-x_j},
\end{equation}
the left hand side of (\ref{thm1}) is equal to the right hand
 side of (\ref{lem1}) and the assertion is proved.
\end{proof}

\subsubsection*{Example 1}
\begin{figure}[h]
\hspace{16mm}
\includegraphics[scale=0.26]{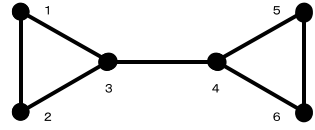} 
\caption{Original graph\label{loopexp}}
\end{figure}

\begin{figure}[h]
\hspace{8mm}
\includegraphics[scale=0.12]{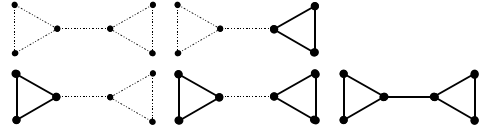} 
\caption{List of generalized loops\label{loopexp}}
\end{figure}

We give an example of the formula given in theorem 2.
Let $G$ be the graph shown in figure 1.
In this case, there are five generalized loops, and the expansion formula
becomes
\begin{equation}
\begin{split}
\frac{Z}{Z_B}
=
1+\beta_{12}\beta_{23}\beta_{13}+\beta_{45}\beta_{56}\beta_{46} 
+\beta_{12}\beta_{23}\beta_{13}\beta_{45}\beta_{56}\beta_{46} \\
\hspace{-5mm}
+\beta_{12}\beta_{23}\beta_{13}\beta_{34}\beta_{45}\beta_{56}\beta_{46}\gamma_{3}\gamma_{4}.
\end{split}
\end{equation}

%%%%%%%%%%%%%%%%%%%%%%%%%%%%%%%%%%%%%%%%
%Marginal 展開公式
%%%%%%%%%%%%%%%%%%%%%%%%%%%%%%%%%%%%%%%
\subsection{Expansion of marginals}
In this subsection, we head for proving theorem 3.
We define a set of polynomials $\{g_n(x)\}_{n=0}^{\infty}$ 
inductively by the relations $g_0(x)=x,g_1(x)=-2$ and $g_{n+1}(x)=x
g_n(x) + g_{n-1}(x)$. This set of polynomials is a transformation of the
Chebyshev polynomials of the first kind.
We introduce the following lemma which is a modification of theorem 1.
\begin{lem}
\begin{equation}
\begin{split}
\sum_{x_1,\ldots,x_N=\pm 1}
\hspace{-2mm}
x_1
\prod_{ij \in E}
(1+x_{i}x_{j}\beta_{ij}\xi_{i}^{-x_i}\xi_{j}^{-x_j})
\prod_{i \in V}
\frac{\xi_{i}^{x_i}}{\xi_{i}+\xi_{i}^{-1}} \\ 
=
\sum_{s \subset E}
\hspace{-0.4mm}
\prod_{ij \in s} \beta_{ij}
\hspace{-0.4mm}
\prod_{i \in V \backslash \{1\}} 
\hspace{-1.2mm}
f_{d_i(s)}(\xi_{i}-\xi_{i}^{-1})
\frac{
g_{d_1(s)}
\hspace{-0.6mm}
(\xi_{1}-\xi_{1}^{-1})
}
{\xi_{1}+\xi_{1}^{-1}}
.
\end{split}\label{lem2}
\end{equation}
\end{lem}

\begin{proof}
Check that
\begin{equation}
\sum_{x_1=\pm 1}x_1 (-x_1 \xi_1^{-x_1})^{n}
\frac{\xi_{1}^{x_1}}{\xi_{1}+\xi_{1}^{-1}}
=
\frac{
g_{n}
(\xi_{1}-\xi_{1}^{-1})
}
{\xi_{1}+\xi_{1}^{-1}}.
\end{equation}
\end{proof}

%判別式公式
\begin{thm} \rm{\cite{watanabe}}  \label{thmmarginal}
\it{Let $p_1(x_1)$ be the true marginal distribution defined by the joint
 probability distribution (\ref{MRF}).
Then,}
\begin{equation}
\hspace{-1mm}
\frac{Z}{Z_B}
\frac{p_{1}(+1)-p_{1}(-1)}{\sqrt{b_{1}(1)b_{1}(-1)}}
\hspace{-1mm}
=
\sum_{s \subset E}
\prod_{ij \in s}
\hspace{-1mm}
\beta_{ij} 
\hspace{-4mm}
\prod_{i \in V \backslash \{1\}} 
\hspace{-3mm}
f_{d_{i}(s)}
\hspace{-0.6mm}
(\gamma_{i})
g_{d_1(s)}
\hspace{-0.6mm}
(\gamma_1). \label{thm3}
\end{equation}
\end{thm}

\begin{proof}
The key fact is
\begin{equation*}
\frac{Z}{Z_B}(p_1(1)-p_1(-1))
\hspace{-0.5mm}
=
\hspace{-0.5mm}
\sum_{\{x_i\}}
x_1
\hspace{-0.5mm}
\prod_{ij \in E}
\hspace{-0.5mm}
\frac{b_{ij}(x_i,x_j)}{b_i(x_i)b_j(x_j)}
\prod_{i \in V}
b_i(x_i),
\end{equation*}
which is a modification of lemma 1.
Follow the proof of theorem 2.
\end{proof}

%1 loop だとexactの話。
A problem of finding an assignment that maximize the marginal probability $p_1$
is called maximum marginal assignment problem.
This problem is especially important in the application of error
correcting code \cite{LDPC,Turbo}.
It is known that if the graph has one cycle and the concerning node 
$1 \in V$ is
on the unique cycle, then the assignment that maximize the belief $b_1$
gives the exact assignment \cite{1loop}.
With theorem 3, this fact is easily deduced as follows \cite{watanabe}.
\begin{cor}
 Let $G$ be a graph with a single cycle with the node $1$ on it.
See figure \ref{figure1loop} for example.
Then,
$p_1(1)-p_1(-1)$ and
$b_1(1)-b_1(-1)$ have the same sign.
\end{cor}
\begin{proof}
In the right hand side of (\ref{thm3}), only two subgraphs $s$ are
 contribute to the sum.
From $g_0(\gamma_1)=\gamma_1$ , $g_2(\gamma_1)=-\gamma_1$ and
 $|\beta_{ij}|\leq 1$, we see that the sum is positively proportional to
 $\gamma_1$.
\end{proof}

%1 loop の絵
\vspace{-3mm}
\begin{figure}[h]
\hspace{25mm}
\includegraphics[scale=0.2]{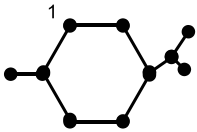} 
\caption{Graph with a single cycle\label{figure1loop}}
\end{figure}
\vspace{-3mm}

We append a comment on theorem 3.
In the sum of (\ref{thm3}), the contribution of $s=\phi$ is equal
to $g_{0}(\gamma_{1})=\gamma_{1}$: this is proportional to
$b_1(+)-b_1(-1)$.
The assignment that maximize the belief $b_1$ can be regarded as the Bethe
approximation in the view point of theorem 3.

%%%%%%%%%%%%%%%%%%%%%%%%%%%%%%%%%%%%%%
%4章
%%%%%%%%%%%%%%%%%%%%%%%%%%%%%%%%%%%%%%%
\section{Bound on the number of generalized loops}
In the loop series expansion formula (\ref{LSE}), the summation runs
over all generalized loops. To know the computational cost for summing
up the terms, the number of generalized loops is of interest. 

\begin{de}
\begin{equation}
\theta_{G}(\beta,\xi):= 
\sum_{s \subset E}
 \beta^{|s|}
\prod_{i \in V} f_{d_i(s)}(\xi-\xi^{-1}).   \label{deftheta}
\end{equation}
\end{de}
\vspace{1mm}
This is a bivariate (Laurent) polynomial with respect to $\beta$ and
$\xi$. The coefficients are all integers.

Let $n(G):=|E|-|V|+1$ be the number of linearly independent cycles.
The next lemma shows that the value of $\theta_{G}$ on $\beta=1$ is
determined by $n(G)$.
%Lemma
\begin{lem}
\begin{equation}
\theta_{G}(1,\xi)=
\sum_{k=0}^{n(G)} {n(G) \choose k}f_{2k}(\xi-\xi^{-1}) \label{lem2}
\end{equation}
\end{lem}

\begin{proof}
The left hand side of theorem 1 gives an alternative representation
 of $\theta_{G}$.
If $x_i \neq x_j$, then $1+x_{i}x_{j}\beta \xi^{-x_i}\xi^{-x_j}=0$.
As we assumed the graph $G$ is connected, only two terms of
 $x_1=\cdots=x_N=1$ and $x_1=\cdots=x_N=-1$ contribute to the sum.
Therefore,
\vspace{-3mm}
\begin{equation}
\begin{split}
\theta_{G}(1,\xi)
=(1+\xi^{-2})^{|E|}
\Big(\frac{\xi}{\xi+\xi^{-1}}\Big)^{|V|} \\
+
(1+\xi^{2})^{|E|}(\frac{\xi^{-1}}{\xi+\xi^{-1}})^{|V|} .\label{lem2p1}
\end{split}
\end{equation}
From (\ref{fn}), the right hand side of (\ref{lem2}) is
 equal to (\ref{lem2p1}).
\end{proof}
%lemma おわり
If $\xi=\frac{1+\sqrt{5}}{2}$, then $\xi-\xi^{-1}=1$. From (\ref{lem2}) or
(\ref{lem2p1}), we see that
\begin{equation*}
\theta_{G}(1,\frac{1+\sqrt{5}}{2})=
 \Biggl(\frac{5-\sqrt{5}}{2}\Biggr)^{n(G)-1}
\hspace{-2mm}
+
\Biggl(\frac{5+\sqrt{5}}{2}\Biggr)^{n(G)-1}.
\end{equation*}
This fact can be used to bound the number of generalized loops \cite{watanabe}.

%Generalized Loops の個数のbound
\begin{thm}
Let $\mathcal{G}_{0}$ be the set of all generalized loops of $G$ including
 empty set. Then,
\begin{eqnarray}
\left| \mathcal{G}_{0} \right|
\leq
 \Biggl(\frac{5-\sqrt{5}}{2}\Biggr)^{n(G)-1}
\hspace{-2mm}
+
\Biggl(\frac{5+\sqrt{5}}{2}\Biggr)^{n(G)-1}.
\end{eqnarray}
This bound is attained if and only if every node of
a generalized loop has the degree at most three.
\end{thm}

\begin{proof}
If we set $\beta=1$ and $\xi=\frac{1+\sqrt{5}}{2}$,
\begin{equation}
\theta_{G}(1,\frac{1+\sqrt{5}}{2})
=
\sum_{s \in \mathcal{G}_0} r(s),
\end{equation}
where $r(s)=\prod_{i \in V}f_{d_{i}(s)}(1)$.
As $f_n(1)>1$ for all $n>4$ and $f_2(1)=f_3(1)=1$,
we have $r(s) \geq 1$ for all $s \in \mathcal{G}_{0}$
, and the equality holds if and only if 
$d_i(C) \leq 3$
for all $i\in V$.
This shows $\left| \mathcal{G}_{0} \right| \leq \theta(1,\frac{1+\sqrt{5}}{2})$ and
the equality condition.
\end{proof}

%%%%%%%%%%%%%%%%%%%%%%%%%%%%%%%%%%%%%%%%%%%%%%%%%%%%%%%%%
%%五章
%%%%%%%%%%%%%%%%%%%%%%%%%%%%%%%%%%%%%%%%%%%%%%%%%%%%%%%%%
\section{Generalization to factor graph model}
In this section we briefly introduce the factor graph model, which is
more general
than the pairwise model in section 1. 
We generalize the result of theorem \ref{thmloop}.
Theorem \ref{thmmarginal} is also generalized straightforwardly.

\subsection{Factor graph model}
Let $H:=(V,F)$ be a hypergraph, that is, $V=\{1,\ldots,N\}$ is a set of
nodes and $F \subset 2^V$ is a set of hyperedges.
A hypergraph $H$ is represented by a bipartite graph $G_{H}=(V_H,E_H)$. Each
type of node corresponds to elements of $V$ and $F$; the first type is
called variable node and the second type is called factor node.
For example,
see figure 4.
In this example, the hypergraph $H=(V,F)$ is given by $V=\{1,2,3\}$ and
$F=\{\lambda_{1},\lambda_{2},\lambda_{3}\}$ ,where
$\lambda_{1}=\{1,2\},\lambda_{2}=\{1,2,3\}$ and $\lambda_{3}=\{2\}$.
\vspace{-3mm}
\begin{figure}[h]
\hspace{25mm}
\includegraphics[scale=0.2]{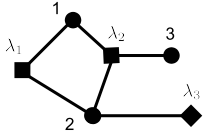} 
\caption{Example of $G_H$\label{figurehypergraph}}
\end{figure}
\vspace{-1mm}

The joint probability distribution is given by the following form:
\begin{equation}
p(x)=\frac{1}{Z}
\prod_{\lambda \in F} \psi_{\lambda}(x_{\lambda}),
\end{equation}
where $x_{\lambda}=\{x_i\}_{i \in \lambda  }$.
If each hyperedge $\lambda \in F$ consists of a pair of nodes, this
class of probability distributions reduces to the pairwise MRF model.

\subsection{Loopy belief propagation algorithm}
For a node $i \in V$ and a hyperedge $\lambda \in F$ which satisfy
$i \in \lambda$, messages $m_{(i,\lambda)}^{}(x_i)$ and
$m_{(\lambda,i)}(x_{\lambda})$
are defined and updated
by the following rules:
\begin{equation*}
\begin{split}
&m^{t+1}_{(i,\lambda)}(x_i)=\omega \sum_{x_{\lambda \setminus \{i\}}}
\hspace{0mm}\psi_{\lambda}(x_{\lambda})
\hspace{-1mm}
\prod_{j \in \lambda, j \neq i}
\hspace{-1mm} 
m^t_{(\lambda,j)}(x_{\lambda}), \\ \label{BPupdatefactor}
&m^{t+1}_{(\lambda,i)}(x_{\lambda})=\omega 
\prod_{\mu \ni i,\mu \neq \lambda} 
\hspace{-1mm}
m^t_{(i,\mu)}(x_{i}).
\end{split}
\end{equation*}

Beliefs are defined by
\begin{equation}
b_{i}(x_i)
:=
\omega
\prod_{\mu \ni i} 
m_{(i,\mu)}^{*}(x_i)
\end{equation}
and
\begin{equation}
b_{\lambda}(x_{\lambda})
:=\omega 
\psi_{\lambda}(x_{\lambda})
\prod_{j \in \lambda } 
m_{(\lambda,j)}^{*}(x_{\lambda}).
\end{equation}

The Bethe approximation of the partition function is given by
\begin{equation}
\begin{split}
\log Z_B := \sum_{\lambda \in {F}}\sum_{x_{\lambda}}
 b_{\lambda}(x_{\lambda})\log\psi_{\lambda}(x_{\lambda})   \\
 - \sum_{\lambda \in {F}}\sum_{x_{\lambda}} b_{\lambda}(x_{\lambda})\log
 b_{\lambda}(x_{\lambda}) \nonumber  \\
+ \sum_{i \in V}(d_i-1)\sum_{x_i}b_i(x_i)\log b_i(x_i). \label{BetheFactor}
\end{split}
\end{equation}

\subsection{Expansion of partition functions}
To state factor graph version of theorem \ref{thmloop}, we need a
little complicated notations.
For each hyperedge $\lambda \in F$ and $I \subset \lambda$, we introduce
a variable $\beta^{\lambda}_{I}$. We use convention that
$\beta^{\lambda}_{\phi}=1$ and $\beta^{\lambda}_{I}=0$ if $|I|=1$.

Theorem 1 is modified to the following identity:
\begin{equation}
\begin{split}
\sum_{\{x_i\}}
\prod_{\lambda \in F}
\sum_{I \subset \lambda}
\hspace{-0mm}
\beta^{\lambda}_{I}
(x_{i_1} \xi_{i_1}^{-x_{i_1}})
\cdots
(x_{i_k} \xi_{i_k}^{-x_{i_k}}) 
\prod_{i \in V}
\frac{\xi_{i}^{x_i}}{\xi_{i}+\xi_{i}^{-1}} \\ 
=
\sum_{s \subset E_{H}}
\prod_{\lambda \in F}
(-1)^{|I_{\lambda}(s)|}
\beta^{\lambda}_{I_{\lambda}(s)}
\prod_{i \in V} f_{d_i(s)}(\xi_{i}-\xi_{i}^{-1}), 
\end{split}\label{thm1mod}
\end{equation}
where $I=\{i_1,\ldots,i_k\}$ and 
$I_{\lambda}(s)$ is a set of variable nodes which connect to $\lambda$ by edges in $s$.

Lemma 1 is modified to
\begin{equation}
\frac{Z}{Z_B} =
\sum_{\{x_i\}}
\prod_{\lambda \in F}
\frac{b_{\lambda}(x_{\lambda})}{\prod_{i \in \lambda}b_{i}(x_i)}
\prod_{i \in V}
b_i(x_i) .
\end{equation}

\begin{thm}
For a factor graph model on $H$, we have
\vspace{0mm}
\begin{equation} \label{LSEfactor}
\hspace{-25mm}
Z=Z_B \sum_{s \subset E_{H}} r(s)
\end{equation}
\begin{equation*}
r(s):=
(-1)^{|s|}
\prod_{\lambda \in F}
\beta^{\lambda}_{I_{\lambda}(s)}
\prod_{i \in V} f_{d_i(s)}(\gamma_{i}).
\end{equation*}
\end{thm}
\begin{proof}[Sketch of proof]
For a given $b_{\lambda}(x_{\lambda})$, we can choose 
$\{\xi_{i}\}_{i \in \lambda}$ and $\{ \beta^{\lambda}_{I}  \}_{I \subset
 \lambda, |I| \geq 2}$ to satisfy
\begin{equation*}
b_{\lambda}(x_{\lambda})=
\frac{1}{\prod_{i \in \lambda} (\xi_{i} + \xi_{i}^{-1})}
\sum_{I \subset \lambda}
\beta^{\lambda}_{I}
\prod_{i \in I} x_i
\prod_{j \in \lambda \setminus I} \xi_{j}^{x_j}.
\end{equation*}
The definition of $\gamma_{i}$ is the same as (\ref{defgamma}).

\end{proof}

In the summation (\ref{LSEfactor}), only generalized loops in $G_H$
contribute to the sum. Therefore, this expansion is again called loop series expansion.

\subsection{Expansion of marginals}
In the same way as section 3.2 we have the following theorem

\begin{thm}
\begin{equation}
\begin{split}
\hspace{-1mm}
&\frac{Z}{Z_B}
\hspace{-0.5mm}
\frac{p_{1}(+1)-p_{1}(-1)}{\sqrt{b_{1}(1)b_{1}(-1)}}\\
\hspace{-1mm}
&=
\hspace{-1mm}
\sum_{s \subset E_H}
\hspace{-1mm}
(-1)^{|s|}
\prod_{\lambda \in F}
\hspace{-0.7mm}
\beta^{\lambda}_{I_{\lambda}(s)}
\hspace{-3mm}
\prod_{i \in V \setminus \{1\}}
\hspace{-2.5mm}
 f_{d_i(s)}\hspace{-0.5mm}(\gamma_{i})
\hspace{1mm}
g_{d_1(s)}
\hspace{-0.5mm}
(\gamma_1). \label{thm6}
\end{split}
\end{equation}
\end{thm}

\section{Properties of $\theta_{G}$}
In the rest of this article, we will focus on the (Laurent) polynomial
of $\theta_{G}$.
The accuracy of the Bethe approximation depends both on the graph
topology and strength of the interactions;
the formula in theorem 2 displays this fact.
To exploit the graph topology for the analysis of the performance of the
LBP algorithm, we need sophisticated techniques.
One of the techniques is graph polynomials.
Graph polynomials have long history since Birkhoff introduced the
chromatic polynomial \cite{Birkhoff} and Tutte generalized it 
to the Tutte polynomial \cite{Tutte}.

\subsection{Contraction-Deletion relation}
In this subsection we explain that the function $\theta_{G}$ admits the
contraction-deletion relation. 
If a function from graph satisfies the contraction-deletion relation and 
multiplicativity for disjoint unions of graphs, 
such a function is called the Tutte's V-function \cite{Tutte}.
Though the Tutte polynomial is the most famous example of the V-functions,
$\theta_{G}$ is another interesting example of the V-functions.

The graph which is obtained by deleting an edge $e \in E$ 
is denoted by $G \backslash e$. 
The graph which is obtained by contracting $e$ is denoted by $G/e$. 
In this article, the operations of the
contraction and the deletion are only applied to the non loop
edges.
Note that an edge $e \in E$ is called a loop if both ends of $e$ are
connected to the same node.

The following formula of $f_n(x)$ is essential in the proof
of the contraction-deletion relation.
\begin{lem}\label{lemf}
 ${}^\forall n,m \in \mathbb{N}$ 

$f_{n+m-2}(x)= f_n(x)f_m(x)+f_{n-1}(x)f_{m-1}(x)$
\end{lem}

%contraction deletion
\begin{thm} \label{thmcd}
For a non loop edge $e \in E$
\begin{equation}
 \theta_{G}(\beta,\xi)=
(1-\beta) \theta_{G\backslash e}(\beta,\xi) +
\beta  \theta_{G/e}(\beta,\xi)
\end{equation}
\end{thm}
\begin{proof}[Sketch of proof]
Classify $s$ in (\ref{deftheta}) if $s$ include $e$ or not.
Then apply lemma $\ref{lemf}$ for $s \ni e$.
\end{proof}

\subsection{The case of $\xi=\sqrt{-1}$ }
At the point of $\xi=\sqrt{-1}$, $\theta_{G}$ has special and interesting properties.
From (\ref{lem2p1}), $\theta_G(1,\sqrt{-1})=0$.
The following theorem asserts that $\theta(\beta,\sqrt{-1})$ can be divided by
$(1-\beta)$ at $|E|-|V|$ times.
\begin{thm} 
\begin{equation}
 \omega_{G}(\beta):= \frac{\theta_{G}(\beta,\sqrt{-1})}{(1-\beta)^{|E|-|V|}} \quad
 \in \mathbb{Z}[\beta]
\end{equation}
\end{thm}
\begin{proof}
Theorem \ref{thmcd} and definition of $\omega_{G}$ imply that
\begin{equation}
\omega_{G}(\beta)=\omega_{G \backslash e}(\beta)+\beta
 \omega_{G/e}(\beta).  \label{thmomega1}
\end{equation}
For a bouquet graph $B_L$, which has a single node and $L$ loops, we
 can easily check that
\begin{equation}
\omega_{B_L}(\beta)=1+(2L-1)\beta.  \label{thmomega2}
\end{equation}
We can show the assertion inductively by (\ref{thmomega1}) and (\ref{thmomega2}).
\end{proof}

The value of $\omega_{G}$ at $\beta=1$ has a combinatorial interpretation.
\begin{thm}
If the graph $G$ does not have loop edge,
\begin{equation*}
\begin{split}
\omega_{G}(1)
%&= \sum_{e \in E}\omega_{G \backslash e}(1) \\
%&= \sum_{s \in \mathcal{H}} 2^k \\
&= \#\{\psi: V \rightarrow E ;\psi \quad \rm{satisfies} \quad \rm{(c1),(c2)} \}.
\end{split}
\end{equation*}
The condition $\rm{(c1)}$ is that $\psi$ is injective and $(\rm{c2})$ is that for all $i \in
 V$ $\psi (i)$ is one of a
connecting edges to $i$.
\end{thm}
We omit the proof of this theorem.
The assumption that the graph is loop-less is not essential.

\subsection{Relation between $\omega_{G}$ and matching polynomials}
The polynomial $\omega_{G}(\beta)$, introduced in the previous section,
is closely related to the matching polynomial.

A matching of $G$ is a set of edges in which any edges does not occupy a
same node. If a matching consists of $k$ edges, it is called k-matching.    
Let $p_{G}(k)$ be the number of k-matchings of $G$. 
The matching polynomial $\alpha_{G}$ is defined by
\begin{equation}
\alpha_{G}(x)= \sum_{k=0}^{[n/2]}(-1)^{k}
p_{G}(k) x^{n-2k}.
\end{equation}

We introduce two square matrices indexed by $V$.
Let $\mathcal{A}$ be a adjacency matrix of the graph $G$ and
let $\mathcal{D}$ be a diagonal matrix called degree
 matrix. 
In other words, for a function $\phi : V \rightarrow \mathbb{R}$,
\begin{equation}
(\mathcal{A} \phi )(i)=\sum_{j \sim i} \phi(j), \quad
(\mathcal{D} \phi )(i)=\textrm{deg}(i) \hspace{1mm}\phi(i).
\end{equation}

%omega多項式の公式
\begin{thm}\label{thmomega}

\begin{equation}
\omega_{G}
(u^{2})
=
\hspace{-3mm}
\sum_{C: cycles}
\hspace{-2mm}
2^{k(C)}
\det[
I
\hspace{-0.5mm}
+
\hspace{-0.5mm}
u^{2}(\mathcal{D}\hspace{-0.5mm}-\hspace{-0.5mm}I)
\hspace{-0.5mm}-u \mathcal{A}
]
\Big|_{G\setminus C}
\hspace{0.8mm}
u^{|C|},
\end{equation}
where
 $\cdot \big|_{G \setminus C}$ denotes a restriction
 to the principal minor of $G \setminus C$.
The summation runs over all node disjoint cycles, i.e.
2-regular edge induced subgraphs.
The number of connected components of $C$ is denoted by $k(C)$
\end{thm}

We omit the proof of theorem \ref{thmomega}.
The following corollary shows that $\omega_{G}$ is a matching polynomial
if $G$ is a regular graph.
\begin{cor}
If $G$ is a $(q+1)$-regular graph, then
\begin{equation}
\omega_{G}(u^{2})
=
\alpha_{G}(1/u +q u)u^{n}.
\end{equation}

\end{cor}
\begin{proof}
In \cite{Godsil}, it is shown that
\begin{equation}
\alpha_{G}
(x)
=
\sum_{C:cycles}
2^{k(C)}
\det
[x I
-
\mathcal{A}_{G \setminus C}
]. \label{mg}
\end{equation}
If $G$ is a $(q+1)$-regular graph, then $\mathcal{D}=(q+1)I$.
From theorem \ref{thmomega} and (\ref{mg}), the assertion follows.
\end{proof}

%%%%%%%%%%%%%%%%%%%%%%%%%%%%%%%%%%%%%%
\section*{Acknowledgments}
This work was supported in part by Grant-in-Aid for JSPS Fellows
20-993 and Grant-in-Aid for Scientific Research (C) 19500249. 

%%%%%%%%%%%%%%%%%%%%%%%%%%%%%%%%%%%%%%

\end{document}